\documentclass[runningheads]{llncs}
\usepackage{latexsym,amssymb,amsmath}
\usepackage{xcolor}

\usepackage{graphicx}
\begin{document}
\title{Free Resolutions and Generalized Hamming Weights of binary linear codes\thanks{This work was partially supported by the Spanish MICINN PID2019-105896GB-I00 and MASCA and MACACO (ULL Research Projects). Third author was supported in part by Grant PGC2018-096446-B-C21 funded by
		MCIN/AEI/10.13039/501100011033 and by “ERDF A way of making Europe”. Fourth author was supported by SEGIB-Fundaci\'on Carolina Grant and SNI M\'exico.}}
%
%
\author{I. Garc\'ia-Marco\inst{1}  \and
I. M\'arquez-Corbella\inst{1}  \and
E. Mart\'inez-Moro\inst{2}\and Y. Pitones\inst{3}}
\authorrunning{I. Garc\'ia-Marco et al.}
%
\institute{Departamento de Matem\'aticas, Estad\'istica e I.O., Universidad de La Laguna,  Tenerife, Spain,
\email{iggarcia,imarquec@ull.edu.es} \and {Institute of Mathematics, University of Valladolid, Castilla, Spain},
\email{edgar.martinez@uva.es}\and {Universidad Aut\'onoma Metropolitana, M\'exico},
\email{ypitones@xanum.uam.mx}
}
\maketitle              
\begin{abstract}
In this work, we explore the relationship between free resolution of some monomial ideals and Generalized Hamming Weights (GHWs) of binary codes. More precisely, we look for a structure smaller than the set of codewords of minimal support that provides us some information about the GHWs. We prove that the first and second generalized Hamming weight of a binary linear code can be computed (by means of a graded free resolution) from a set of monomials associated to a binomial ideal related with the code. Moreover, the remaining weights are bounded by the Betti numbers for that set. 

\end{abstract}
\section{Introduction}

The study of the Generalized Hamming Weights (GHWs) has been motivated by several applications in cryptography \cite{W:91} and they  characterize the performance of a linear code when it is used for a given channel. There are few families of codes for which it is known the complete generalized weight hierarchy as for example: first-order Reed-Muller codes, binary Reed-Muller codes, Hamming code and its dual, Extended Hamming codes, Golay code, see \cite{W:91}. On the other side,  there has been an extensive research on GHWs and the second distance in particular for some  classes of codes, see for example \cite{GHW1,GHW5,GHW9,GHW8,GHW2,GHW6,GHW10,GHW7,GHW4,GHW3} and the references therein. However, for the general case of a  linear code, few properties are known. 

In their seminal paper \cite{JV:2013}, Johnsen and Verdure showed how the GHWs of a linear code could be computed from a free resolution of a monomial ideal associated to the set of codewords of minimal support of the code provided we know this last set. That paper has produced a great avenue of research, see for example \cite{JVrel6,GS:2020,JV:2014,JVrel7,JVrel5,JVrel4,JVrel3,JVrel2,JVrel1}. 

In the present work,  we explore if one can find a set (smaller than  the set of all codewords of minimal support) that provides us some information on the GHWs in the case of binary codes. The selected set of codewords is the so called Gr\"obner test set related to the binomial ideal associated to a code defined in \cite{BBFM:2008}. In that paper it was proven that one can decode using that set and that the minimal distance of the code (thus the first GHW) was reflected on it. Thus, somehow some of the relevant information of the code lies in it. Moreover, in \cite{BBFM:2008} it was also shown how to compute the  Gr\"obner test set avoiding some of the most common disadvantages when one uses Gr\"obner basis. In this paper we will show how one can compute also the  second GHW of a binary linear code   from that binomial ideal associated with  the code without the need of computing the complete set of codewords of minimal support of the code as in \cite{JV:2013}.  Moreover, in Theorem~\ref{teo:gordo} we bound the remaining GHWs with the resolution of a monomial ideal associated to this new set.  

The structure  of the contribution is simple, in Section~\ref{sec:intro} we revise some results on GHW, free resolutions and the ideal associated with a code. Section~\ref{sec:conjfalse} shows briefly the false conjecture and experiments that drive us to this study. 
In Section~\ref{sec:main} we show our main result and finally in Section~\ref{sec:last} we show some future lines of research and some conjectures related to the topic that we hope will be helpful for future research.

\section{GHW and minimal supports}\label{sec:intro}

Let $\mathbb F_q$ be a finite field with $q$ elements. Given two vectors $\mathbf x = (x_1, \ldots, x_n)$ and $\mathbf y = (y_1, \ldots, y_n)\in \mathbb F_q^n$, the \emph{Hamming distance} between $\mathbf x$ and $\mathbf y$ is defined as
$$d_H(\mathbf x, \mathbf y) = |\{ i \mid x_i \neq y_i\}|,$$
where $|\cdot|$ denotes the cardinality of the set. The \emph{Hamming weight} of $\mathbf x$ is given by $\mathrm{w}_H(\mathbf x) =d_H (\mathbf x, \mathbf 0)$, where $\mathbf 0$ denotes the zero vector in $\mathbb F_q^n$. The \emph{support} of $\mathbf x$ is the set $\mathrm{supp}(\mathbf x) = \{ i \mid x_i \neq 0\}$. 
A linear subspace $\mathcal C$ in $\mathbb F_q^n$ is called a \emph{linear code}. The elements of $\mathcal C$ are called \emph{codewords}. The \emph{basic parameters} of $\mathcal C$ are its length, its  dimension and its minimum distance, which are denoted by $n(\mathcal C)$, $k(\mathcal C)$ and $\delta(\mathcal C)$, respectively. In this case, we call $\mathcal C$ an $[n(\mathcal C), k(\mathcal C), \delta(\mathcal C)]_q$ linear code. We define a \emph{generator matrix} of $\mathcal C$ to be a  matrix $G$ over $\mathbb F_q$ of size $k(\mathcal C) \times n(\mathcal C)$ whose row vectors span $\mathcal C$, while a \emph{parity check matrix} of $\mathcal C$ is a matrix $H$ over $\mathbb F_q$ of size $(n(\mathcal C) - k(\mathcal C)) \times n(\mathcal C)$ whose null space is $\mathcal C$.

{
\begin{definition}
Let $\mathcal C$ be a linear code, we will say that a codeword $\mathbf m$ has minimal support if it is non-zero and $\mathrm{supp}(\mathbf m)$ is not contained in the support of any other codewords. We will denote by $\mathcal M_{\mathcal C}$ the set of codewords of minimal support of $\mathcal C$.
\end{definition}
}

{Note that computing a set of codewords of minimal support is a hard problem for a general linear code, in fact it implies finding the minimum distance of the code and, thus, it is harder that the problem of maximum-likelihood complete decoding, see \cite{Hard,B98}.}

\begin{definition}
	Let $\mathcal C$ be a linear code and $D$ be a subcode of $\mathcal C$, we define \emph{the support of $D$}, denoted $\mathrm{supp}(D)$ as the set of not-always-zero bit positions of $D$, i.e.,
	$$\mathrm{supp}(D) = \{ i \mid \exists \mathbf c \in D \hbox{ with } c_i \neq 0\}.$$
\end{definition}

It is clear that  if $D$ is a one-dimensional subcode then, the support of $D$ is equal to the Hamming weight of any of its   nonzero codewords, i.e.  $d_1(\mathcal C) = \delta (\mathcal C)$. Based on this idea, the \emph{$h$-th generalized Hamming weight}  of $\mathcal C$, denoted $d_h(\mathcal C)$, is the size of the smallest support of an $h$-dimensional subcode of $\mathcal C$ with $h=1, 2, \ldots, k(\mathcal C)$. That is, if $D_h$ is the set of all linear subspaces of the linear code $\mathcal C$ of dimension $h$, then
$$d_h (\mathcal C) = \min \{ |{\rm supp} (E)| \mid E \in D_h\}.$$
Some basic facts on the Generalized Hamming Weights (GHW) are provided in the following proposition.
\begin{proposition}[\cite{W:91}]
	Let $\mathcal C$ be a linear code. Then:
	\begin{enumerate}
		\item $1\leq d_1(\mathcal C) < d_2(\mathcal C) < \ldots < d_{k(\mathcal C)}(\mathcal C) \leq n(\mathcal C)$
		\item  {(Generalized Singleton Bound)} $d_h(\mathcal C) \leq n(\mathcal C) - k(\mathcal C) + h$.
	\end{enumerate}
\end{proposition}
From now on we assume that   $\mathcal C$ is a nondegenerate code, that is,   $d_{k(\mathcal C)}(\mathcal C)~=~n(\mathcal C)$.

The GHWs $d_1(\mathcal C)$, \ldots, $d_{k(\mathcal C)}(\mathcal C)$ are completely determined by the underlying linear matroid structure of the code in a non-trivial manner, and the way to obtain them is not efficient as it was already pointed in \cite{JV:2013}. Of course, as pointed before, calculating $d_1(\mathcal C)$  is equivalent to the problem of complete decoding linear codes~\cite{B98}, hence  one can not expect a computational efficient approach.
From now on, given a positive integer $\ell,$ we define $\left[\ell\right]=\left\{1,\ldots,\ell\right\}$ and $\left[\ell\right]_0=\left\{0,\ldots,\ell\right\}$. 

\begin{definition}
Let $\mathcal C$ be a $[n,k]_q$ code and let $H$ be a parity check matrix of $\mathcal C$. Let $H_i$ denote the $i$-th column of $H$ and define the simplicial complex
$$\Delta= \left\{ \sigma \in 2^{\left[n\right]} \mid \{ H_i \mid i \in \sigma\} \hbox{ is linearly independent over } \mathbb F_q\right\}.$$
Then, the pair $\mathcal M= (\left[n\right], \Delta)$ is the linear \emph{matroid associated with the code} $\mathcal C$. The collection $\Delta$ of subsets of $\left[n\right]$ are called \emph{independent sets} of this matroid. A subset of $\left[n\right]$ that does not belong to $\Delta$ is called \emph{dependent set}. Minimal dependent subsets of $\left[n\right]$ are known as \emph{circuits} of $\mathcal M$. A set is said to be a \emph{cycle} if it is a disjoint union of circuits.
\end{definition}
We refer to \cite{ES:2005} for a brief introduction on the theory of simplicial complexes, and to \cite{OX} for a thorough study of matroids. 



\begin{definition}
Let $\mathbb K$ be any field. We denote by $I_{\Delta}$ the ideal in  the polynomial ring $R = \mathbb K [X_1, \ldots, X_n]$ over a field $\mathbb K$ generated by all square-free monomials supported on elements that are not in $\Delta$, i.e.
$$\prod_{i \in \tau} X_i \hbox{ with } \tau \in 2^{\left[n\right]} \setminus \Delta$$
That is, $I_{\Delta}$ is the ideal generated by monomials supported on the circuits of $\mathcal M$, or equivalently in a matroid associated with a code, supported on codewords of minimal support of $\mathcal C$.
\begin{equation}
I_{\Delta} = \left\langle \prod_{i \in \mathrm{supp}(\mathbf c)} X_i \mid \mathbf c \in \mathcal M_{\mathcal C}\right\rangle.
\label{ideal-matroid}
\end{equation}
\end{definition}
The quotient $R_{\Delta} = R / I_{\Delta}$ is called the \emph{Stanley-Reisner} ring associated with $\Delta$. $R_{\Delta}$ is a finitely generated standard graded $\mathbb K$-algebra of dimension $n(\mathcal C)- k(\mathcal C)$. Thus, it has a minimal graded free resolution. Moreover, since the generators of $I_\Delta$ are supported in the set of circuits of a matroid, by \cite{B:1992}, one has that $\Delta$ is shellable and this implies that $R_{\Delta}$ is Cohen-Macaulay. So, by the Auslander-Buchsbaum formula, the projective dimension of $R_{\Delta}$ (i.e., the length of any minimal graded free resolution of $R_{\Delta}$) is $k(\mathcal C)$ and it looks like
\begin{equation}
\label{MFR-eq}
\begin{array}{ccccccccccccccc}
0 & \longrightarrow & F_{k(\mathcal C)} & \longrightarrow & F_{k(\mathcal C)-1} & \longrightarrow & \cdots & \longrightarrow & F_1 & \longrightarrow & F_0 & \longrightarrow & R_{\Delta} & \longrightarrow 0
\end{array}
\end{equation}
where $F_0 = R$ and each $F_i$ is a graded free $R$-module of the form
$$F_i = \bigoplus_{j \in \mathbb N} R(-j)^{\beta_{i,j}} \hbox{ for } i \in \left[k(\mathcal C)\right]_0.$$
We will refer to Equation~\eqref{MFR-eq} as a graded minimal free resolution of $\mathcal C$.
The nonnegative integers $\beta_{i,j}$ are called \emph{Betti numbers} of $\mathcal C$ and they depend only on $\mathcal C$ and not on the choice of the parity check matrix $H$, or the minimal free resolution of $R_{\Delta}$ or  even the chosen  field $\mathbb K$ (this is because $I_\Delta$ comes from the set of circuits of a matroid, see \cite{B:1992}). 

In \cite{JV:2013}, Johnsen and Verdure described the GHWs of a linear code $\mathcal C$ in terms of the shifts of the minimal graded free resolution of a Stanley-Reisner ideal $I_\Delta$ associated to $\mathcal C$. One of the main disadvantages of this method is that the generators of $I_\Delta$ correspond to the supports of all minimal support of $\mathcal C$. In general, the whole set of codewords of minimal support can be huge and computationally expensive to obtain and, in many cases, computing a minimal graded free resolution of an ideal with that many generators is unaffordable. More precisely, their result was:

\begin{theorem} \cite[Theorem 2]{JV:2013}
\label{Verdure-Teo}
Let $\mathcal C$ be a $q$-ary linear code. Then,
$$d_i(\mathcal C) = \min \{ j \mid \beta_{i,j} \neq 0\} \hbox{ for } j \in \left[k(\mathcal C)\right]$$
\end{theorem}

{
\begin{example}[Toy example]
\label{toy-example1}
Let $\mathcal C$ be the binary non-degenerate $[6,3]$ code with generator matrix
$$G=\left( \begin{array}{cccccc}
1 & 0 & 0 & 0 & 0 & 1\\
0 & 1 & 1 & 0 & 1 & 0 \\
0 & 0 & 0 & 1 & 1 & 1
\end{array}\right) \in \mathbb F_2^{3\times 6}$$
One can check that the set of codewords of minimal support is
$$\mathcal M_{\mathcal C} = \left\{ \begin{array}{ccc}
w_1= 100001, & w_2=100110, & w_3=011010, \\
w_4= 000111, & w_5 = 111100, & w_6= 011101 
\end{array}
\right\}$$
Its Stanley-Reisner ring is $R_{\Delta} = R / I_{\Delta}$ where $R=\mathbb F_2[x_1, \ldots, x_6]$ and the ideal $I_{\Delta}$ is generated by the monomials associated to $\mathcal M_{\mathcal C}$. If we compute a graded minimal free resolution of $R_{\Delta}$ we get the following Betti diagram:
\[
\begin{array}{c|cccc}
& 0 & 1 & 2 & 3  \\ \hline 
0 &  1 & 0 & 0 & 0 \\
1 & 0 & 1 & 0 & 0 \\
2 & 0 & 3 & 2 & 0\\
3 & 0 & 2 & 7 & 4 \\
\end{array}
\]
where the entry of the row indexed by $i$ and column indexed by $j$ indicates the value $\beta_{j,i+j}$ (for example $\beta_{3,5} = 7$).
Hence, by Theorem \ref{Verdure-Teo}, we have $d_1(\mathcal C) = 2$, $d_2(\mathcal C) = 4$ and $d_3(\mathcal C) = 6$.
\end{example}
}

This result shows that the determination of the Betti numbers of the monomial ideal  $I_{\Delta}$  related to a code completely determine the weight hierarchy. However, as mentioned before, this is usually a hard problem \cite{B98} except in some special cases. 
For example, Johnsen and Verdure in \cite{JV:2014} explicitly determine the Betti Numbers for MDS codes, since the minimal free resolution of these codes is linear. 
Moreover, in \cite{JV:2014} the authors prove that the resolution of the first order Reed-Muller code is pure. And a similar result can be deduced for constant weight codes \cite{GS:2020}. Thus, simplex codes or dual Hamming codes, which are constant weight codes also have pure resolution, although not necessarily linear.

\begin{remark}The resolution \eqref{MFR-eq} is said to be \emph{pure} of type $(d_0, \ldots, d_{k(\mathcal C)})$ if for each $i\in \left[k\right]_0$, the Betti number $\beta_{i,j}$ is nonzero if and only if $j= d_i$. If, in addition $d_1, \ldots, d_k$ are consecutive, then the resolution is said to be \emph{linear}. 
\end{remark}

The rest of this work is devoted to compute a simpler and smaller structure than the whole set of codewords of minimal support that allows, at least partially, know something about the GHW of the code.

\subsection{Test-Sets of a binary code}

From now on, we will restrict our study to binary codes. 
Let $\mathcal C$ be a binary linear code. Let $X$ be a vector with $n = n(\mathcal C)$  variables $x_1, \ldots, x_{n}$. A monomial in $X$ is a product of the form $X^{\mathbf a} := \prod_{i = 1}^n x_i^{\,a_i},$ where $\mathbf a = (a_1, \ldots, a_{n}) \in \mathbb N^{n}$. The total degree of $X^{\mathbf a}$ is $\deg(X^{\mathbf a}) = \sum_{i=1}^{n} a_i$. When $\mathbf a = \mathbf 0$, note that $X^{\mathbf a} = \mathbf 1$. The polynomial ring $\mathbb K[X]$ is the set of all polynomials in $X$ with coefficients in $\mathbb K$, where $\mathbb K$ denotes an arbitrary finite field.

\begin{remark}
For abuse of notation we will write $X^{\mathbf a}$ with $\mathbf a \in {\mathbb F_2^{n}}$. In this case, we understand that the classes of $0,1$ are replaced by the same symbols regarded as integers. Moreover,  we will use the notation $X^I$ for the square free monomial with support $I \subset [n]$, that is, 
$$X^I = \prod_{i \in I} x_i \hbox{ with } I \subseteq \left[ n \right].$$
Let $g=X^A - X^B$ be a binomial with $A,B\subseteq \left[ n\right]$, we define $\mathrm{supp}(g) = A \cup B$. We say that $g$ is in \emph{standard form} if $A\cap B = \emptyset$ or, equivalently, if $\mathrm{supp}(g) = A \sqcup B,$ where $\sqcup$ denotes the disjoint union of $A$ and $B$.
\end{remark}

\begin{definition}
Let $\mathbb K$ be any field, we define the \emph{ideal associated with $\mathcal C$ over $\mathbb K$} as the binomial ideal:
\begin{equation}
I(\mathcal C) = \langle  X^{\mathbf a} - X^{\mathbf b} \mid \mathbf a, \mathbf b \in \mathbb F_2^n,\, \mathbf a + \mathbf b \in \mathcal C  \rangle + \langle x_i^2 - 1 \mid i \in [ n ] \rangle \subseteq \mathbb K[X].
\label{ideal-code}
\end{equation}
\end{definition}
Note that $I(\mathcal C)$ is a zero-dimensional ideal since the quotient ring $R= \mathbb K[X] / I(\mathcal C)$ is a finite-dimensional vector space (i.e. $\dim_{\mathbb K}(R)< \infty$). Moreover, its dimension is equal to the number of \emph{cosets} in $\mathbb F_2^{n} / \mathcal C$.
For $\mathbf a, \mathbf b \in \mathbb F_2^n$ one has that $X^{\mathbf a}  - X^{\mathbf b} \in I(\mathcal C)$ if and only if $\mathbf a - \mathbf b \in \mathcal C$. 
The following result shows how to obtain a set of generators of the ideal $I(\mathcal C)$ from a generator matrix of $\mathcal C$.

\begin{proposition}\cite[Theorem 1]{BBFM:2008}
Let $\{\mathbf w_1, \ldots, \mathbf w_k\}$ be the row vectors of a generator matrix for $\mathcal C$. Then
$$I(\mathcal C) = \left\langle \left\{ X^{\mathbf w_i}-1\right\}_{i \in \left[ k \right]} \cup \left\{ x_i^2 -1\right\}_{i \in \left[ n \right]} \right\rangle$$
\end{proposition}

If we fix a term order $\prec$, then the \emph{leading term} of a polynomial $f$ with respect to $\prec$, denoted by $\mathrm{LT}_{\prec}(f)$, is the largest monomial among all monomials which occur with non-zero coefficient in the expansion of $f$. Let $I$ be an ideal in $\mathbb K[X]$, then the \emph{initial ideal} $\mathrm{in}_{\prec}(I)$ is the monomial ideal generated by the leading term of all the polynomials in $I$, i.e., $\mathrm{in}_{\prec}(I) = \langle \left\{ \mathrm{LT}_{\prec}(f) \mid f\in I\right\} \rangle$. By definition, $\mathrm{in}_{\prec}(I)$ is a monomial ideal and, thus, it has a unique minimal generating set formed by monomials. These monomials will be called {\it minimal generators of} $\mathrm{in}_{\prec}(I)$.

\begin{definition}
A finite set of nonzero polynomials $\mathcal G=\{ g_1, \ldots, g_m\}$ of the ideal $I$ is a Gr\"obner basis of $I$ with respect to the term order $\prec$ if the leading terms of the elements of $\mathcal G$ generate the initial ideal $\mathrm{in}_{\prec}(I)$. Moreover $\mathcal G$ is reduced if
\begin{enumerate}
\item $g_i$ is monic (i.e., its leading coefficient is $1$) for all $i \in \left[m\right]$, and
\item none of the monomials appearing in the expansion of $g_j$ is divisible by $\mathrm{LT}_{\prec}(g_i)$ for all $i \neq j$.
\end{enumerate}
\end{definition}

For a given monomial order $\prec$, every ideal has a unique reduced Gr\"obner basis (see, e.g., \cite{AL}).
Since $I(\mathcal C)$ is generated by binomials (differences of monomials), then all its reduced Gr\"obner bases consist of binomials (see \cite{ES:96}). In \cite{BBFM:2008} it is shown that, if  $\mathcal{C}$ is a binary code and we fix a degree compatible term order $\prec$ on $\mathbb K[X]$, then the  reduced Gröbner basis $\mathcal G_\prec$ for the code ideal $I(\mathcal C)$ can be computed by a linear algebra (an FGLM-like) algorithm. Moreover, the reduction provided by $\mathcal G_\prec$ gives a decoding procedure. Along the way, they also prove that the support of every binomial in $\mathcal G_\prec$ different from $x_i^2 - 1$ for $i = 1,\ldots,n$ provides a codeword of minimal support of $\mathcal C$, and that there is a word of Hamming weight $d_1(\mathcal C)$ that can be obtained in this way. More precisely:

\begin{proposition}\cite{BBFM:2008}\label{pr:BBFM} Let $\mathcal G_\prec$ be the reduced Gr\"obner basis of $I(\mathcal C)$ with respect to a degree compatible term order $\prec$. For every binomial $X^{\mathbf a} - X^{\mathbf b} \in \mathcal G_\prec - \{x_i^2 - 1\, \vert \, i \in [n]\}$, then $\mathbf a + \mathbf b \in \mathbb F_2^n$ is a codeword of minimal support of $\mathcal C$. Moreover, there exists $X^{\mathbf a} - X^{\mathbf b} \in \mathcal G_\prec$ such that $\mathrm{w}_H(\mathbf a + \mathbf b ) = d_1(\mathcal C)$. 
\end{proposition}
This result motivates the definition of a $\mathcal G_{\prec}$-test, which by the above proposition is a subset of the set of codewords of minimal support and contains a word of minimum weight. 

\begin{definition}
Given a binary code $\mathcal C$ and a degree compatible term order $\prec$, we will call $\mathcal G_\prec$-test set of $\mathcal C$ to the subset of codewords of minimal support of $\mathcal C$ whose supports are given by those binomials in the reduced Gr\"obner basis of $I(\mathcal C)$ different from $x_i^2-1$
for all $i \in [n]$.
\end{definition}

{
\begin{example}
\label{toy-example2}
We continue the toy Example  \ref{toy-example1}. Then, the ideal associated with $\mathcal C$ over $R$ is defined as
$$I(\mathcal C) = \left\langle \left\{\begin{array}{ccc}
x_1x_6-1, &x_2x_3x_5-1, &x_4x_5x_6-1
\end{array}\right\} \cup \left\{ x_i^2-1 \right\}_{i \in \left[ 6 \right]}\right\rangle \subseteq R$$
Now we consider the degree reverse lexicographic order $\prec$ with $x_6\prec\ldots \prec x_1$. Then, the reduced Gr\"obner basis $\mathcal G_{\prec}$ of $I(\mathcal C)$ with respect to $\prec$ has $14$ elements. An a $\mathcal G_{\prec}$-test-set of $\mathcal C$ is given by codewords of $\mathcal M_{\mathcal C}$ whose supports are given by those binomials of $\mathcal G_{\prec}$ different from $x_i^2-1$, i.e. the set
$$ \left\{ \begin{array}{cccc}
w_1=100001, & w_3 = 011010, & w_4 = 000111, & w_6=011101
\end{array}\right\}\subseteq \mathcal M_{\mathcal C}.$$

\end{example}
}

\section{On a false conjecture}\label{sec:conjfalse}
 Note that the second and the third author of this work have conjectured for a long time the following: \\[0.025em]
 
\noindent \textit{ \textbf{Conjecture (false)} If one considers the monomial ideal $M$ associated with the supports of the binomials in the $\mathcal G_\prec$-test set, then $d_i = {\rm min}\{j \, \vert \, \beta_{i,j}(R/M) \neq 0\}$ for $i \in \{1,\ldots,k(\mathcal C)\}$, that is, the $\mathcal G_\prec$-test set determines the GHWs of the codes.}\\[0.025em]

  \noindent This conjecture was supported by  computational evidence (see Examples \ref{toy-example3} and \ref{ex:ej1} for some example satisfying this conjecture),  but unfortunately this is not true as shown in the counterexample~\ref{ex:contraej}. Theorem \ref{teo:gordo} in Section~\ref{sec:main} in this paper  will prove this fact for $i = 2$ (and it was known for $i = 1$ \cite{BBFM:2008}). Note also that in \cite{Irene}, the authors show that from the Graver basis associated to $I(\mathcal C)$ one can purge the set of codewords of minimal support of $\mathcal C$, i.e. $\mathcal M_{\mathcal C}$.

	\begin{example}
		\label{toy-example3}
		We continue the Example \ref{toy-example1} where we have computed the Betti diagram associated to $R_{\Delta} = R / I_{\Delta}$ where $I_{\Delta}$ is generated by the monomials associated to $\mathcal M_{\mathcal C}$. Then, in Example \ref{toy-example2} we compute a $\mathcal G_{\prec}$-test-set $\mathcal T_{\mathcal C}$ of $\mathcal C$ with respect to the degree reverse lexicographic order $\prec$. Note that $\mathcal T_{\mathcal C}\subset \mathcal M_{\mathcal C}$ with just $4$ elements. If we consider $M$ the corresponding monomial ideal related to $\mathcal T_{\mathcal C}$ and compute a graded minimal free resolution of $R / M$ we get the following Betti diagram:
		\[
		\begin{array}{c|cccc}
		& 0 & 1 & 2 & 3  \\ \hline 
		0 &  1 & 0 & 0 & 0 \\
		1 & 0 & 1 & 0 & 0 \\
		2 & 0 & 2 & 1 & 0\\
		3 & 0 & 4 & 4 & 2 \\
		\end{array}
		\]
		The Betti numbers of $R/M$ are smaller than those of $R_{\Delta}$ and in this example the sequence
		\[ ({\rm min}\{ j \, \vert \, \beta_{i,j}\neq 0 \}; 1 \leq i \leq 3) = (2,4,6), \]
		coincides with the GHW of $\mathcal C$
	\end{example}
	
	The following example is a less trivial case where the difference between our structure $\mathcal T_{\mathcal C}$ and the set of codewords of minimal support of $\mathcal C$, $\mathcal M_{\mathcal C}$, is larger.

\begin{example}\label{ex:ej1} Let $\mathcal C$ be the binary non-degenerate $[14,9]$-code with generator matrix
	\[G = \left(
	\begin{array}{cccccccccccccccccccccc}
	1 & 1 & 0 & 1 & 0 & 0 & 0 & 0 & 0 & 1 & 0 & 0 & 0 & 0\\
	1 & 1 & 1 & 1 & 0 & 1 & 0 & 0 & 1 & 1 & 0 & 0 & 0 & 1\\
	0 & 0 & 1 & 0 & 0 & 1 & 1 & 1 & 1 & 1 & 0 & 1 & 1 & 0\\
	1 & 1 & 0 & 1 & 1 & 1 & 0 & 0 & 1 & 1 & 0 & 0 & 0 & 1\\
	0 & 1 & 1 & 1 & 1 & 0 & 0 & 1 & 0 & 0 & 1 & 0 & 1 & 1\\
	0 & 0 & 1 & 1 & 0 & 1 & 1 & 0 & 1 & 0 & 0 & 1 & 0 & 1\\
	0 & 0 & 1 & 0 & 1 & 1 & 1 & 0 & 1 & 1 & 1 & 0 & 1 & 0\\
	0 & 1 & 0 & 1 & 1 & 0 & 1 & 0 & 1 & 1 & 0 & 1 & 1 & 0\\
	1 & 1 & 1 & 1 & 1 & 1 & 0 & 1 & 0 & 0 & 0 & 0 & 1 & 0\\
	\end{array} \right) \in \mathbb F_2^{\,9 \times 14}\]
	
	Its Stanley-Reisner ring is $R_{\Delta} = R/I_{\Delta}$ where $R=\mathbb F_2[x_1, \ldots, x_{14}]$ and $I_\Delta$ is minimally generated by $147$ monomials. If we compute a graded minimal free resolution of $R_{\Delta}$ we get the following Betti diagram:
	
	\[
	\begin{array}{c|ccccccccccc}
	& 0 & 1 & 2 & 3 & 4 & 5 & 6 & 7 & 8 & 9 \\ \hline 
	0 &     1   &  0 &    0  &   0  &   0  &   0   &  0  &   0     & 0    & 0 \\
	1 &     0   &  2 &    0  &   0  &   0  &   0   &  0  &   0     & 0    & 0 \\
	2 &     0   &  8 &    5  &   0  &   0  &   0   &  0  &   0     & 0    & 0 \\
	3 &     0   & 34 &   82  &  48  &   8  &   0   &  0  &   0     & 0    & 0 \\
	4 &    0   & 52 &  441  & 897  & 753  & 289   & 42  &   0     & 0    & 0 \\
	5 &     0   & 51 & 1345  & 7410 &18309 & 25248 & 21008 & 10579 & 2990 &  366
	
	\end{array}
	\]
	Hence, by Theorem \ref{Verdure-Teo}, we have that the GHW are $$(d_1(\mathcal C),\ldots,d_9(\mathcal C)) = (2,4,6,7,9,10,12,13,14).$$
	Moreover, if we consider $\prec$ the graded degree lexicographic order with $x_1 \succ \cdots \succ x_{14}$, we get that the $\mathcal G_\prec$-test set has 24 elements. If we consider $M$ the corresponding monomial ideal and compute a graded minimal free resolution of $R / M$ we get the following Betti diagram:
	
	\[
	\begin{array}{c|ccccccccccc}
	& 0 & 1 & 2 & 3 & 4 & 5 & 6 & 7 & 8 & 9 \\ \hline 
	0 &    1   &  0  &   0  &   0 &    0   &  0   &  0   &  0  &   0    & 0 \\
	1 &    0   &  2  &   0  &   0 &    0   &  0   &  0   &  0  &   0    & 0 \\
	2 &    0   &  6  &   3  &   0 &    0   &  0   &  0   &  0  &   0    & 0 \\
	3 &    0   & 13  &  38  &  17 &    2   &  0   &  0   &  0  &   0    & 0  \\
	4 &    0   &  3  &  92  & 194 &  130   &  35  &   3  &   0 &    0    & 0 \\
	5 &    0   &  0  &  83  & 599 & 1410   & 1621 & 1040 &  378 &    71  &   5 \\
	6 &    0   &  0  &   0  &   0 &    2   &   5  &   4  &   1  &    0   &  0 \\

	\end{array}
	\]
	As one can observe, the Betti numbers of $R/M$ are smaller than those of $R_\Delta$. Moreover, in this example the sequence 
$ ({\rm min}\{ j \, \vert \, \beta_{i,j}\neq 0 \}; 1 \leq i \leq 9)$ is given by $ (2,4,6,7,9,10,12,13,14)$ and it 
	coincides with the GHW of $\mathcal C$.
\end{example}

\begin{example}\label{ex:contraej}
	Let $\mathcal C$ be the binary non-degenerate $[10,7]$-code with generator matrix
	$$G = \left(\begin{array}{cccccccccc}
	1 & 0 & 0 & 0 & 1 & 0 & 0 & 1 & 1 & 1 \\
	1 & 1 & 1 & 1 & 0 & 1 & 1 & 1 & 1 & 1 \\
	1 & 0 & 0 & 1 & 0 & 1 & 1 & 0 & 0 & 0 \\
	0 & 0 & 1 & 1 & 1 & 0 & 1 & 0 & 0 & 1 \\
	1 & 0 & 1 & 0 & 0 & 1 & 1 & 0 & 0 & 1 \\
	0 & 0 & 1 & 1 & 0 & 1 & 1 & 1 & 0 & 0 \\
	0 & 0 & 1 & 1 & 0 & 0 & 0 & 1 & 1 & 1 \\
	\end{array}\right) \in \mathbb F_2^{\,7\times 10}$$

	Its Stanley-Reisner ring is $R_{\Delta} = R/I_{\Delta}$ where $R=\mathbb F_2[x_1, \ldots, x_{10}]$ and $I_\Delta$ is minimally generated by $42$ monomials. If we compute a graded minimal free resolution of $R_{\Delta}$ we get the following Betti diagram:
	$$
	\begin{array}{c|cccccccc}
	& 0 & 1 & 2 & 3 & 4 & 5 & 6 & 7 \\ \hline
	0 & 1 & 0 & 0 & 0 & 0 & 0 & 0 & 0\\
	1 & 0 & 4 & 0 & 0 & 0 & 0 & 0 & 0\\
	2 & 0 & 18 & 48 & 32 & 7 & 0 & 0 & 0\\
	3 & 0 & 20 & 214 & 637 & 874 & 637 & 242 & 38\\
	
	\end{array}
	$$-/
	Hence, by Theorem \ref{Verdure-Teo}, we have $d_1(\mathcal C) = 2$, $d_2(\mathcal C) = 4$, $d_3(\mathcal C) = 5$, $d_4(\mathcal C) = 6$, $d_5(\mathcal C) = 8$, $d_6(\mathcal C) = 9$, $d_7(\mathcal C) = 10$.
	
	If we compute a $\mathcal G_\prec$-test set $T$   with respect to the degree reverse lexicographical order, we get that $T$ has only 10 elements. Computing a minimal graded free resolution of $R/M$ being the ideal $M := \langle \{ X^{\mathbf c} \, \vert \, \mathbf c \in T \} \rangle$,   we get the following Betti diagram.
	$$
	\begin{array}{c|ccccccc}
	& 0 & 1 & 2 & 3 & 4 & 5 & 6 \\ \hline
	0 & 1 & 0 &  0 & 0 & 0 & 0 & 0 \\
	1 & 0 & 4 & 0 & 0 & 0 & 0 & 0 \\
	2 & 0 & 4 & 14 & 5 & 0 & 0 & 0 \\
	3 & 0 & 2 & 23 & 56 & 48 & 17& 2 \\
	\end{array}
	$$
	
	Thus, we can recover the correct values of $d_1(\mathcal C), d_2(\mathcal C), d_3(\mathcal C)$ from this resolution but not $d_i(\mathcal C)$ for $i =4,5,6,7$.
\end{example}
Anyway, the experiments we dealt trying to prove that conjecture drove us the results in Section~\ref{sec:main} and to propose some further conjectures and future work based on experimental  evidences showed in Section~\ref{sec:last}.

\section{Second GHW obtained from a $\mathcal G_\prec$-test set}\label{sec:main}  

In this section we will explain how to compute the second generalized Hamming weight from a $\mathcal G_\prec$-test set. Throughout this section, $\mathcal C$ will denote   a binary linear code and let $\mathcal G_{\prec}$ be the reduced Gr\"obner basis of the ideal $I(\mathcal C)$ with respect to $\prec$, where we take $\prec$ to be any degree compatible ordering on $\mathbb K[X]$.
The following result is a technical lemma whose proof is just an easy exercise in set theory. We will denote by $A\triangle B$ the \emph{symmetric difference} of the subsets $A, B \subseteq [n]$, which is the set of elements which are in the union of the two sets $A\cup B$, minus their intersection $A\cap B$. Moreover, for a set $A$, its cardinality is denoted by $|A|$.

\begin{lemma}
\label{SymetricDifLemma}
Let $A,B \subseteq X$ with $|A\cap B|> \frac{|A|}{2}$. Then $C=A\triangle B$ satisfies the following statements: 
\begin{enumerate}
\begin{multicols}{3}
    \item $A\cup B=A\cup C$
    \item $|A\cap C| < \frac{|A|}{2}$
    \item $|C|< |B|$.
\end{multicols}
\end{enumerate}
\end{lemma}
 
Consider the set $\mathbf M$ of codewords in $\mathcal C$ belonging to a linear subspace of dimension two of minimal support, more precisely, 
\begin{equation}
 \mathbf M = \{\mathbf m \in \mathcal C \, \vert \, \exists \mathbf m' \in \mathcal C \text{ such that }
d_2(\mathcal C) = \mathrm{supp}(\left \langle \mathbf m, \mathbf m' \right\rangle)\}, 
\label{eq:minimal} \end{equation} 
and define $\mathbf m_1, \mathbf m_2 \in \mathbf M$ as follows:
\begin{itemize}
\item[(a)] $\mathbf m_1 := {\rm min}_{\prec}(\mathbf M)$, i.e., $\mathbf m_1$ is the smallest codeword with respect to $\prec$ in $\mathbf M$, 
\item[(b)] $\mathbf m_2 := {\rm min}_{\prec}\{\mathbf m \in \mathbf M \, \vert \, d_2(\mathcal C) = \mathrm{supp}(\left \langle \mathbf m_1, \mathbf m \right)\rangle\}$, i.e, $\mathbf m_2$ is the smallest codeword with respect to $\prec$ such that 
$d_2(\mathcal C) = \mathrm{supp}(\left \langle \mathbf m_1, \mathbf m_2 \right\rangle)$.
\end{itemize}
With these conditions we define:
$$\begin{array}{ccc}
I = \mathrm{supp}(\mathbf m_1) & \hbox{ and } &
J= \mathrm{supp}(\mathbf m_2).
\end{array}$$

\begin{remark}\label{remark1}
Since $\prec$ is degree compatible and, by Lemma \ref{SymetricDifLemma}, we have that $$|I\cap J|\leq \frac{|I|}{2}\leq \frac{|J|}{2}.$$
\end{remark}

\begin{proposition}\label{Prop1}
There exists a binomial $f\in \mathcal G_{\prec}\subseteq \mathbb K[X]$ such that $\mathrm{supp}(f) = I$.
\end{proposition}

\begin{proof}
Consider $f=X^{I_1}-X^{I_2} \in I(\mathcal C)$ any binomial with $I=I_1 \sqcup I_2$, $|I_1| - 1 \leq |I_2| \leq |I_1|$ and $X^{I_1}\succ X^{I_2}$, we will show that $f\in \mathcal G_{\prec}$. For proving this, it suffices to check that:
\begin{itemize}
\item[(a)] $X^{I_1} = \mathrm{LT}_{\prec}(f)$ is a minimal generator of $\mathrm{in}_{\prec}(I(\mathcal C))$, and
\item[(b)] $X^{I_2} \notin \mathrm{in}_{\prec}(I(\mathcal C))$.
\end{itemize}

Proof of (a). By contradiction suppose that there exists a binomial $h = X^{K_1}- X^{K_2}\in \mathcal G_{\prec}$ with 
$K=K_1 \sqcup K_2$, 
$\mathrm{LT}_{\prec}(h) = X^{K_1}$ (and, in particular, $|K_1| \geq |K_2|$) such that $K_1 \neq I_1$ and $X^{I_1}$ is divisible by $X^{K_1}$, i.e. $K_1 \subsetneq I_1$. 

\smallskip

\noindent \emph{Claim:} $d_2(\mathcal C) \leq |I\cup K| -1$. 

\noindent \emph{Proof of the claim:} We have that $|K_2| \leq |K_1|  \leq |I_1| - 1 \leq |I_2|$. Then $|K| < |I|$ and, in particular, $X^K \prec X^I \prec X^J$. Hence, by the choice of $\mathbf m_2$, we have that $|I \cup  J| < |I \cup K|$ and the \emph{Claim} follows.

By the previous \emph{Claim} we have that
\begin{eqnarray*}
|I \cup J| & = & d_2(\mathcal C) \leq |I\cup K| -1 ~~~~ \\
&=& |I\cup K_2|-1 ~~~~ \hbox{(since } K_1 \subsetneq I \hbox{)} \\
&=& |I| + |K_2| - |I\cap K_2|-1\\ 
&\leq &|I| + |K_1|- |I\cap K_2|-1 ~~~~ \hbox{(since } X^{K_1} \succ X^{K_2} \hbox{)}.
\end{eqnarray*}
Thus, 
$$|I| + |J| - |I \cap J| = |I \cup J| \leq |I| + |K_1|- |I\cap K_2|-1$$
or equivalently $|J| - |I \cap J| \leq |K_1|- |I\cap K_2|-1$. Thus,
\begin{eqnarray*}
|I_1| - \frac{1}{2}  & \leq & \frac{|I|}{2} \leq  \frac{|J|}{2} \leq |J| - |I\cap J| ~~~~ \hbox{ (by Remark \ref{remark1})}\\
& \leq & |K_1|- |I\cap K_2|-1 \leq |K_1|-1.
\end{eqnarray*}
Therefore, $|I_1| \leq |K_1| - \frac{1}{2}$ and $K_1\subsetneq I_1$, a contradiction.

Proof of (b). By contradiction, we assume that there exists $h= X^{K_1} - X^{K_2}\in \mathcal G_{\prec}$ such that $\mathrm{LT}_{\prec}(h)=X^{K_1}$ divides $X^{I_2}$ or, equivalently, $K_1 \subseteq I_2$. Then $K_2 \prec K_1 \prec I_2 \prec I_1$ and, in particular, $X^K \prec X^I \prec X^J$.  Hence, by the choice of $\mathbf m_2$, we have that $|I \cup  J| < |I \cup K|.$
\begin{eqnarray*}
d_2(\mathcal C) &=& |I \cup J| \leq |I\cup K|-1 = |I \cup K_2|-1 = |I| + |K_2|-|I\cap K_2|-1 \\
&\leq& |I| + |K_2|-1 \leq |I|+ |K_1|-1.
\end{eqnarray*}
Therefore, $|J|-|I\cap J|\leq |K_1|-1$. And, by Remark \ref{remark1} we deduce that 
$$|I_2|\leq \frac{|I|}{2}\leq \frac{|J|}{2}\leq |J|-|I\cap J| \leq |K_1| - 1.$$
Therefore, $|I_2|<|K_1|$ and $K_1\subseteq I_2$, which is a contradiction. \qed
\end{proof}


One can check that the same result (and the same proof) holds for all $I'$ such that $I' = {\rm supp}(\mathbf m)$ with $\mathbf m \in \mathbf M$ and $|I'| = |I|$. As a consequence of this observation we have that:
\begin{corollary}
If $|I| = |J|$. Then, we can always find binomials $f,g\in \mathcal G_{\prec}\subseteq \mathbb K[X]$ such that $\mathrm{supp}(f) = I$ and $\mathrm{supp}(g) = J$.
\end{corollary}
 We just found that $I$ is always involved in the supports associated with $\mathcal G_{\prec}$ and that if $J$ has the same cardinal as them, also the second GHW can be derived from the Gr\"obner basis. Let us prove now the general case.
 
 \begin{proposition}\label{Prop2}
There exists a binomial $g\in \mathcal G_{\prec}\subseteq \mathbb K[X]$ such that  $\mathrm{supp}(g) = J$.
 \end{proposition}
 
 \begin{proof} By Remark \ref{remark1} we know that $|I \cap J| \leq |J|/2$, then one may consider $$g=X^{I\cap J}X^{J_1}-X^{J_2} \in I(\mathcal C ), $$
 such that $J_1\cup J_2=J-I$, $J_1\cap J_2=\emptyset$ and $|J_2| + 1 \geq |I \cap J| + |J_1| \geq |J_2|$. Our goal is to prove that $g$ (or $-g$) is in $\mathcal G_{\prec}$. We split the proof in two:
 \begin{description}
\item[Case ${\rm LT}_{\prec}(g)= X^{I\cap J}X^{J_1}$.] We are going to see that $g \in \mathcal G_{\prec}$.    
First, we show that $X^{I\cap J}X^{J_1}$ is a minimal generator of ${\rm in}_\prec (I(\mathcal{C}))$. By contradiction, suppose that there is a binomial $h=X^{K_1}-X^{K_2}\in \mathcal G_\prec$ with 
$\mathrm{LT}_{\prec}(h) = X^{K_1}$ such that $K_1 \subsetneq (I\cap J)\cup J_1$, and denote $K=K_1 \sqcup K_2$.
We have that $|J_2| \geq |I \cap J| + |J_1| - 1 \geq |K_1| \geq |K_2|$ and, in particular, $|J| > |K|$. 

%
%
%
%

However, 
\[ \begin{array}{lcl}
d_2(\mathcal C) & = &  |I\cup J|=|I|+|J|-|I\cap J|\\
 & = & |I|+ |J_1|+|J_2| \\
 & \geq & |I|+|J_1|+|K_2|\\
 & \geq & |I\cup K|, \end{array}\]
which cannot happen by the choice of $J = {\rm supp}(\mathbf m_2)$. Therefore, $X^{I \cap J} X^{J_1}$ is a minimal generator of ${\rm in}_\prec (I(\mathcal{C}))$. 
 
 Now, we will show that $X^{J_2}\notin {\rm in}_{\prec}(I(\mathcal{C}))$. 
 Since $X^{I\cap J}X^{J_1}$ is a minimal generator of ${\rm in}_{\prec}(I(\mathcal{C}))$, then there exists $h=X^{I\cap J}X^{J_1}-X^{L}\in \mathcal G_{\prec}$. We have that  $X^L \notin {\rm in}_{\prec}(I(\mathcal{C}))$ and let us see that $L = J_2$. Suppose that $L \neq J_2$, then  $0 \neq h - g = x^{J_2} - x^L \in I(\mathcal C)$ and $X^{J_2} \succ X^L$, because $X^L \notin {\rm in}_{\prec}(I(\mathcal{C})).$ However, \[ |I \cup {\rm supp}(h)| \leq |I| + |J_1| + |L| \leq |I| + |J_1| + |J_2| = |I \cup J| = d_2(\mathcal C),\] which is a contradiction.

\item[Case ${\rm LT}_{\prec}(g)=  X^{J_2}$.] We are going to see that $-g \in \mathcal G_{\prec}$.  First, we show that $X^{J_2}$ is a minimal generator of ${\rm in}_\prec (I(\mathcal{C}))$. Suppose that there exists a binomial $h=X^{K_1}-X^{K_2} \in\mathcal G_\prec$ with $X^{K_1} \succ X^{K_2}$ and $K_1 \subsetneq J_2$. Consider $L := J\triangle K$, we have that $L \subseteq (J - K_1) \cup K_2$ and, hence, $X^L \prec X^{J - K_1} X^{K_2} \prec X^J$. However, \[ |I \cup L| \leq |I \cup (J - K_1) \cup K_2| \leq |I \cup J| - |K_1| + |K_2| \leq |I \cup J| = d_2(\mathcal C),\] and again it is
a contradiction.
\end{description}
 
 We will show now that $X^{I\cap J} X^{J_1}\notin \mathrm{in}_\prec (I(\mathcal{C}))$, by contradiction. Suppose that $X^{I\cap J} X^{J_1}\in \mathrm{in}_\prec (I(\mathcal{C}))$, then there exists a binomial $X^{I\cap J} X^{J_1}- X^{L} $ in $I(\mathcal{C})$. Consider now $K=(I\cap J)\cup J_1\cup L$, again one can compute that $|I\cup J|\geq |I\cup K|$ and that $K\prec J$, which again contradicts the choice of $\mathbf m_2$.

 \qed
 \end{proof}
 From the above results the main theorems of this paper will follow.
 
 \begin{theorem}\label{Thmdeglex}
Let $\prec$ a degree compatible order  
in $\mathbb K[X]$, then there are $f, g$ in the reduced Gr\"obner basis $\mathcal G_\prec$ of $I(\mathcal C)$ with respect to $\prec$, such that $d_2=|{\rm supp}(f) \cup {\rm supp}(g)|$. 
\end{theorem}

\begin{example}\label{ex:ejth2} Let $\mathcal C$ be the binary non-degenerate $[6,3]$-code with generator matrix
\[G = \left(
\begin{array}{cccccccccccccccccccccc}
1 & 0 & 0 & 1 & 1 & 0 \\
0 & 1 & 0 & 1 & 0 & 1 \\
0 & 0 & 1 & 0 & 1 & 1
\end{array} \right) \in \mathbb F_2^{\,3 \times 6}\]

One can check that every codeword different from $0$ is a codeword of minimal support and, hence, there are $7$ codewords of minimal support, namely $w_1 = 111000, w_2 = 110011, w_3 = 101101, w_4 = 100110, w_5 = 101110, w_6 = 010101, w_7 = 001011$. Moreover, we have that 
$d_1(\mathcal C) = 3$ and $d_2(\mathcal C) = 5$ and $d_3(\mathcal C) = 6$ and the set $\mathbf M$ described in (\ref{eq:minimal}) coincides with $\mathcal C - \{0\}$.

Consider the degree reverse lexicographic order $\prec_1$ with $x_6 \prec_1 \cdots \prec_1 x_1$, then \begin{itemize} \item $\mathbf m_1 := \min_{\prec_1}(\mathbf M) = 001011 = w_7$, and  \item $\mathbf m_2 := {\rm min}_{\prec_1}\{\mathbf m \in \mathbf M \, \vert \, d_2(\mathcal C) = \mathrm{supp}(\left \langle \mathbf m_1, \mathbf m \right)\rangle\} = 010101 = w_6.$ \end{itemize} The reduced Gr\"obner basis $\mathcal G_1$ of $I(\mathcal C)$ with respect to $\prec_1$ has 20 elements. As proved in Proposition \ref{Prop1}, the binomial $f = x_5 x_6 - x_3$ belongs to $G_1$ and has ${\rm supp}(f) = {\rm supp}(w_7) = \{3,5,6\}$ and, by Proposition \ref{Prop2}, the binomial $g = x_4 x_6 - x_2$ belongs to $\mathcal G_1$ and has  ${\rm supp}(g) = {\rm supp}(w_6) = \{2,4,6\}$. Moreover, $d_2 = |{\rm supp}(f) \cup {\rm supp}(g)| = |{\rm supp}(\langle w_6,w_7\rangle)| = |\{2,3,4,5,6\}| = 5$.

If we consider the degree reverse lexicographic order $\prec_2$ with $x_1 \prec_2 \cdots \prec_2 x_6$, then \begin{itemize} \item $\mathbf m_1' := \min_{\prec_2}(\mathbf M) = 111000 = w_1$, and  \item $\mathbf m_2' := {\rm min}_{\prec_2}\{\mathbf m \in \mathbf M \, \vert \, d_2(\mathcal C) = \mathrm{supp}(\left \langle \mathbf m_1', \mathbf m \right)\rangle\} = 100110 = w_4.$ \end{itemize} The reduced Gr\"obner basis $\mathcal G_2$ of $I(\mathcal C)$ with respect to $\prec_2$ has 20 elements. As proved in Proposition \ref{Prop1}, the binomial $f' = x_1 x_2 - x_3$ belongs to $ \mathcal G_2$ and has ${\rm supp}(f) = {\rm supp}(w_1) = \{1,2,3\}$ and, as proved in Proposition \ref{Prop2}, the binomial $g' = x_1 x_4 - x_5$ belongs to $\mathcal G_2$ and has  ${\rm supp}(g) = {\rm supp}(w_4) = \{1,4,5\}$. Moreover, $d_2 = |{\rm supp}(f') \cup {\rm supp}(g')| = |{\rm supp}(\langle w_1,w_4\rangle)| = |\{1,2,3,4,5\}| = 5$.

\end{example} 
As a consequence, we have that the two first GHWs $d_1(\mathcal C)$ and $d_2(\mathcal C)$ can be obtained from the  minimal graded free resolution  associated with the supports in the $\mathcal G_\prec$-test set.  
Moreover, from this resolution one  can also obtain upper bounds for all the  $d_i(\mathcal C)$ where  $i \in \left[ k(\mathcal C) \right]$. More precisely, we have the following.

\begin{theorem}\label{teo:gordo} Let $\mathcal C \subseteq \mathbb F_2^n$ be a binary code, $\prec$ a degree compatible monomial order in $R$. Let $T$ denote the $\mathcal G_\prec$-test, define the square-free monomial ideal \[ M := \langle \{ X^{\mathbf c} \, \vert \, \mathbf c \in T \} \rangle  \subseteq R,\] and consider  
 a minimal graded free resolution of $R/M$: 
\[ 0 \longrightarrow F_p \longrightarrow \cdots \longrightarrow  F_2  {\longrightarrow}  F_1 \longrightarrow R \longrightarrow R/M \longrightarrow 0,  \]
where each $F_i$ is a graded free $R$-module of the form
$$F_i = \bigoplus_{j \in \mathbb N} R(-j)^{\beta_{i,j}(R/M)} \hbox{ for } i \in \left[p\right]_0.$$
Then, 
\begin{itemize}
\item[{\rm (a)}] $p \leq k(\mathcal C)$
\item[{\rm (b)}] $d_i(\mathcal C) \leq {\rm min}\{j \, \vert \, \beta_{i,j}(R/M) \neq 0\}$ for all $j \in \{3,\ldots,p\}$, and 
\item[{\rm (c)}] $d_i(\mathcal C) = {\rm min}\{j \, \vert \, \beta_{i,j}(R/M) \neq 0\},$ for $i = 1, 2$.
\end{itemize}
\end{theorem}
\begin{proof} By Proposition \ref{pr:BBFM}, every $\mathbf c \in T$ is a codeword of minimal support and, hence, $\{X^{\mathbf c} \, \vert \, \mathbf c \in T\}$ is a subset of the generators of the monomial ideal supported on all the codewords of minimal support. Thus, (a) and (b) follow from  Theorem \ref{Verdure-Teo}.

The rest of the proof concerns (c). Again by Proposition \ref{pr:BBFM}, $\{X^{\mathbf c} \, \vert \, \mathbf c \in T\}$ is the minimal monomial generating set of $M$. Since $\beta_{1,i}(R/M)$ equals the number of minimal generators of $M$ of degree $i$, then, 
$\beta_{1,i}(R/M) = | \{\mathbf{c} \in T \, \vert \, \mathrm{w}_H(c) = i\}. |$ 
By Proposition \ref{pr:BBFM}, there is a $\mathbf c \in T$ such that $\mathrm{w}_H(\mathbf c) = d_1$ and, thus, \[d_1(\mathcal C) = {\rm min}\{ \mathrm{w}_H(\mathbf c) \, \vert \, \mathbf c \in T\} = {\rm min}\{i \, \vert \, \beta_{1,i}(R/M) \neq 0\}.\]

Assume now that $T = \{\mathbf c_1,\ldots,\mathbf c_r\},$ and consider $\mathcal T$ the Taylor resolution of $M = \langle X^{\mathbf c_1},\ldots,X^{\mathbf c_r} \rangle$. The first steps of this resolution are given by 
$$\mathcal T: \cdots \longrightarrow  F_2'  \overset{\varphi_2}{\longrightarrow}  F_1' \longrightarrow R \longrightarrow R/M \longrightarrow 0,$$
where $F_i' := \oplus_{|I| = i \atop I \subset [r]}  R (-|{\rm supp}\langle \mathbf c_j\, \vert \, j \in I \rangle|)$ and $F_1' := \oplus_{1 \leq i \leq r}  R (-|{\rm supp}(\mathbf c_i)|).$ 
Hence, the shifts in the second step of $\mathcal T$ are given by $|{\rm supp}\langle \mathbf c_i, \mathbf c_j \rangle|$ for $1 \leq i < j \leq r$ and, as a consequence, $d_2(\mathcal C) \leq {\rm min}\{|{\rm supp}\langle \mathbf c_i, \mathbf c_j \rangle|\, \vert \, 1  \leq i < j \leq r\}.$ Moreover, by Theorem \ref{Thmdeglex}, this is indeed an inequality. 

In general, the Taylor resolution is not minimal (it is usually very far from minimal). However, it can be pruned to get a minimal one. Consider now $$\mathcal F: \cdots \longrightarrow  F_2  \longrightarrow  F_1 \longrightarrow R \longrightarrow R/M \longrightarrow 0,$$
a minimal graded free resolution of $R/M$ obtained after pruning the Taylor one. As we said before, $\{X^{\mathbf c_1},\ldots,X^{\mathbf c_r}\}$ is the minimal monomial generating set of $M$ and, hence, $F_1' = F_1$. As a consequence, ${\rm min}\{i \, \vert \, \beta_{2,i}(R/M) \neq 0\} = {\rm min}\{|{\rm supp}\langle \mathbf c_i, \mathbf c_j \rangle| \, \vert \, 1 \leq i < j \leq r\} = d_2(\mathcal C).$
\qed
\end{proof}

\section{Final remarks, future work \& new conjectures}\label{sec:last}

 In this work, we build on the results of Borges-Quintana et al. \cite{BBFM:2008} and propose $\mathcal G_\prec$-test sets as a smaller structure from where one can obtain the values of $d_1(\mathcal C)$ and $d_2(\mathcal C)$ for binary codes.
Several experiments with SageMath \cite{sagemath} suggest that Theorem \ref{teo:gordo} can also be extended for $i = 3$. More precisely:

\begin{question}
Let $\mathcal C \subseteq \mathbb F_2^n$ be a binary code, $\prec$ a degree compatible monomial order in $R$. Consider the $\mathcal G_\prec$-test set $T$ and define the square-free monomial ideal \[ M := \langle \{ X^{\mathbf c} \, \vert \, \mathbf c \in T \} \rangle  \subseteq R.\] Is
$d_3(\mathcal C) = {\rm min}\{i \, \vert \, \beta_{3,i}(R/M) \neq 0\}$?
\end{question}

In Example \ref{ex:contraej} one has that the projective dimension (i.e., the number of steps of the resolution) of $R/M$ is ${\rm pd}(R/M) = 6$, while the dimension of $\mathcal C$ is $k(\mathcal C) = 7$. In all the counterexamples to the original conjecture that we have found it turns out that ${\rm pd}(R/M) < k(\mathcal C)$. This motivates us to ask if the conjecture holds provided that ${\rm pd}(R/M) = k(\mathcal C)$. More precisely:

\begin{question}
Whenever ${\rm pd}(R/M) = k(\mathcal C)$, is it true that \[ d_i(\mathcal C) = {\rm min}\{j \, \vert \, \beta_{i,j}(R/M) \neq 0\} \] for all $i \in \{1,\ldots,k(\mathcal C)\}$?
\end{question}

Also, the following natural questions arise.
\begin{question} What is in between the test set and the complete set of codewords of minimal support? i.e. Can we characterize a mid-way structure that provides the complete set of GHWs? 
\end{question} 

A possible candidate fot that intermediate set could be the union of all $\mathcal G_\prec$-test sets for all $\prec$ degree compatible orderings. In general, this set can be smaller than the whole set of codewords of minimal support { and can be computed by the algorithm proposed in \cite{Natalia}}. For example, for the $[7,4]$ binary Hamming code, i.e., the code with  generator matrix

\[ \left( \begin{matrix}
1 & 0 & 0 & 0 & 0 & 1 & 1 \\
0 & 1 & 0 & 0 & 1 & 0 & 1 \\
0 & 0 & 1 & 0 & 1 & 1 & 0 \\
0 & 0 & 0 & 1 & 1 & 1 & 1 \\
\end{matrix} \right) \in \mathbb F_2^{\,4 \times 7},
\]
has $14$ codewords of minimal support, half of them with  weight $3$, and the rest with   weight $4$. Moreover, the sequence of GHW is $(d_1(\mathcal C),\ldots,d_4(\mathcal C)) = (3,5,6,7)$. One has that all the $\mathcal G_\prec$-test sets when $\prec$ ranges over all degree compatible orderings consists of the $7$ codewords of Hamming weight $3$. If one computes the Betti diagram of the monomial ideal $M \subset \mathbb F_2[x_1,\ldots,x_7]$ corresponding to this set, one gets the following

$$
\begin{array}{c|ccccccc}
& 0 & 1 & 2 & 3 & 4  \\ \hline
0 & 1 & 0 & 0 & 0 & 0 \\
1 & 0 & 0 & 0 & 0 & 0\\
2 & 0 & 7 & 0 & 0 & 0 \\
3 & 0 & 0 & 21 & 21 & 6 \\
\end{array}
$$
Hence, the sequence 
$ ({\rm min}\{ j \, \vert \, \beta_{i,j}\neq 0 \}; 1 \leq i \leq 4) = (3,5,6,7)$
coincides with the sequence of GHWs of the code $\mathcal C$.

\begin{question} Can we say something in the non-binary case?  \end{question}
	In order to answer this question, one could try to apply the techniques in \cite{Irene}, where a generalization of the ideal $I(\mathcal C)$ for non-binary codes is studied. 	

	
	
\begin{question} Recently,  E. Gorla and A. Ravagnani in \cite{GR22} extend and generalize the results of Johnsen and Verdure \cite{JV:2013} to compute the generalized weights of a code with respect to a different notions of weight. It would be interesting to see if these generalized weights can be computed from a $\mathcal G_\prec$-test set of $\mathcal C$.  \end{question}

\section*{Acknowledgments}

We would like to thank E. Gorla (University of Neuchatel, Switzerland) for her valuable comments and suggestions.

%
%
%
\bibliographystyle{splncs04}
\bibliography{seconddistance}
\end{document}